\documentclass[a4paper, notitlepage, 11pt]{article}
\usepackage[utf8]{inputenc}
\usepackage[english]{babel}
\usepackage[T1]{fontenc}
\usepackage{amsmath,amsthm,amssymb,amscd}
\usepackage{color}
\usepackage{soul}
\usepackage{url}
\usepackage[normalem]{ulem}
\usepackage{enumitem}
\usepackage[left=2.5cm,right=2.5cm,bottom=2.5cm,top=2.5cm]{geometry}
\usepackage{tikz}
\usepackage{appendix}
\usepackage{csquotes}
\usepackage{float}
\usepackage{fancyhdr}
\usepackage{mathtools}
\usepackage{stmaryrd}
\usepackage[sort&compress,square,numbers]{natbib}%this replaces [1,2,3,4] by [1-4].
\usepackage{diagbox}
\usetikzlibrary{patterns}

\newcommand{\bra}[1]{\langle #1|}
\newcommand{\cket}[1]{|#1\rangle}
\newcommand{\bracket}[2]{\langle #1|#2\rangle}
\newcommand{\mcl}[1]{\mathcal{#1}}
\newcommand{\norme}[2]{\left|\left|#1\right|\right|_{#2}}

\newcommand{\Acal}{{\mathcal A}}
\newcommand{\Bcal}{{\mathcal B}}

\newcommand{\EcalKDC}{{\mathcal E}_{\mathrm{KD+}}}
\newcommand{\EcalKDCpu}{{\mathcal E}_{\mathrm{KD+}}^{\mathrm{pure}}}
\newcommand{\EcalKDCext}{{\mathcal E}_{\mathrm{KD+}}^{\mathrm{ext}}}

\newcommand{\convAB}{\conv{\Acal \cup \Bcal}}

\newcommand{\Hcal}{{\mathcal H}}

\newcommand{\Kext}{{K_{\mathrm{ext}}}}

\newcommand{\Ncal}{{\mathcal N}}
\newcommand{\PiAcal}{{\Pi_{\Acal}}}
\newcommand{\PiBcal}{{\Pi_{\Bcal}}}

\newcommand{\N}{\mathbb{N}}

\newcommand{\R}{\mathbb{R}}

\newcommand{\scr}{s^{\textsc{cr}}}
\newcommand{\Ncr}{\mcl{N}^{\textsc{cr}}}

\newcommand{\nab}{n_{\Acal,\Bcal}}
\newcommand{\nabcr}{n_{\Acal, \Bcal}^{\textsc{cr}}}

\newcommand{\na}{n_{\Acal}}
\newcommand{\nb}{n_{\Bcal}}
\newcommand{\mab}{m_{\Acal,\Bcal}}

\newcommand{\Tr}{\mathrm{Tr}\,}
\newcommand{\conv}[1]{\mathrm{conv}\left(#1\right)}

\newcommand{\id}[1]{\ensuremath{{I}_{#1}}}

\renewcommand{\H}{\mathcal{H}}

\newcommand{\IntEnt}[2]{\llbracket #1 , #2 \rrbracket}
\renewcommand{\epsilon}{\varepsilon}

\newcommand{\Minnabcr}{\mcl{M}_{\nabcr}}

	\renewcommand{\thesection}{\arabic{section}}
	\renewcommand{\thesubsection}{\arabic{section}.\arabic{subsection}}

\newtheorem{Theorem}{Theorem}[section]

\newtheorem{defini}[Theorem]{Definition}

\newtheorem{Lemma}[Theorem]{Lemma}
\newtheorem{Prop}[Theorem]{Proposition}

\title{Convex roofs witnessing Kirkwood-Dirac nonpositivity }
\author{Christopher Langrenez$^1$\thanks{christopher.langrenez@univ-lille.fr}, David R.M. Arvidsson-Shukur$^2$\thanks{drma2@cam.ac.uk}, Stephan De Bi\`evre$^1$\thanks{stephan.de-bievre@univ-lille.fr}\\
$\,^1$Univ. Lille, CNRS, Inria, UMR 8524, Laboratoire Paul Painlev\'e, F-59000 Lille, France\\
$\,^2$ Hitachi Cambridge Lab., J.J Thomson Avenue, Cambridge CB3 0HE, UK 
}

\begin{document}

\maketitle

\begin{abstract}

Given two observables $A$ and $B$, one can associate to every quantum state a Kirkwood-Dirac (KD) quasiprobability distribution. KD distributions are like joint classical probabilities except that they can have negative or nonreal values, which are associated to nonclassical features of the state. In the last decade, KD distributions have come to the forefront as a versatile tool  to investigate and construct  quantum advantages and nonclassical phenomena. KD distributions are also used to determine quantum-classical boundaries. To do so, one must have witnesses for when a state is KD nonpositive. Previous works have established a relation between the 
uncertainty of a pure state with respect to the eigenbases of $A$ and $B$ and  KD positivity. If this \textit{support uncertainty} is large, the state cannot be KD positive. Here, we construct two witnesses for KD nonpositivity for general mixed states. Our first witness is the convex roof of the support uncertainty; it is not faithful, but it extends to the convex hull of pure KD-positive states the relation  between KD positivity and small support uncertainty. Our other witness is the convex roof of the total KD nonpositivity, which provides a faithful witness for the convex hull of the pure KD-positive states. This implies that the convex roof of the total nonpositivity captures the nonpositive nature of the KD distribution at the underlying pure state level.
\end{abstract}

\tableofcontents
\section{Introduction}

Quasiprobability distributions play an important role in the development of quantum mechanics and in quantum information processing. They may be used to distinguish states which can manifest strong quantum properties from those that cannot. The most commonly encountered  quasiprobability distributions are the Wigner distribution and the Glauber-Sudarshan P-function. These distributions  are defined for systems characterized by two conjugate variables $X$ and $P$. Wigner-positive states are states with an everywhere nonnegative Wigner function. Such states  are known not to provide a quantum advantage in quantum computation~\cite{MariEisert2012}. In quantum optics, states with a positive P-function are often referred to as ``classical'' since they mimic in many ways the behaviour of classical light~\cite{titulaer_correlation_1965}. We shall find it convenient to call such states P positive. The P-positive states form a subset of the Wigner-positive states. In both cases, much effort has gone into finding witnesses, measures, and monotones of negativity; see~~\cite{hillery_classical_1985,bach_simplex_1986,hillery_nonclassical_1987,hillery_total_1989,lee_measure_1991,agarwal_nonclassical_1992,lee_theorem_1995,lutkenhaus_nonclassical_1995,dodonov_hilbertschmidt_2000,marian_quantifying_2002,richter_nonclassicality_2002,kenfack_negativity_2004,asboth_computable_2005,ryl_unified_2015,sperling_convex_2015,killoran_converting_2016, yadin_general_2016, nair_nonclassical_2017, ryl_quantifying_2017, alexanian_non-classicality_2018,  yadin_operational_2018, kwon_nonclassicality_2019,de_bievre_measuring_2019,luo_quantifying_2019, Hoetal19, bohmann_probing_2020,  hertz_quadrature_2020, hertz_relating_2020, Arnhem2022, hertzdebievre2023}.
Equivalently, this amounts to obtaining a convenient and precise description of all Wigner-positive or P-positive states.  The question is whether there are  witnesses that can determine if a state $\rho$ is Wigner positive and/or P positive? For pure states, the answer is well known. By Hudson's theorem, the only pure Wigner-positive states are the pure Gaussian states~\cite{hudson1974}. Those are precisely the pure states that minimize the Robertson-Schr\"odinger uncertainty relation: $\det\gamma\geq 1/4$ (See~\cite{serafini2017}). Here $\gamma$ is the covariance matrix of the state:
\begin{equation}\label{eq:covmatrix}
\gamma=\begin{pmatrix}
    \sigma_x^2&\sigma_{xp}^2 \\ \sigma_{xp}^2 & \sigma_p^2
\end{pmatrix},
\end{equation}
with $\sigma_x^2$ and $\sigma_p^2$ being the variances of $X$ and  $P$, and 
\begin{equation}
    \sigma_{xp}^2=\frac12\langle XP+PX\rangle- \langle X\rangle\langle P\rangle.
\end{equation}
It follows from the definition of the P function that the only pure P-positive states are the coherent states, which minimize the additive uncertainty relation $\sigma_x^2+\sigma_p^2\geq 1$. For mixed states, the question is considerably more complex. 

In this paper, we address analogous questions for the Kirkwood-Dirac (KD) distribution, the definition of which we first recall.
The definitions of the Wigner function and of the P function are intimately related to the existence of conjugate variables for the systems considered. The KD distributions, on the other hand, form  a versatile family of quasiprobability distributions well adapted to discrete variable systems. They  are defined as follows~\cite{kirkwood1933, dirac1945, arvidssonshukur2024properties}. Given two orthonormal bases  $ \left(\cket{a_i}\right)_{i\in\IntEnt{1}{d}}$ and $\left(\cket{b_j}\right)_{j\in\IntEnt{1}{d}}$ in a complex Hilbert space $\H$ of dimension $d$, the KD-distribution $Q(\rho)$ of a state $\rho$ is given by
\begin{equation}\label{eq:KDGV1}
\forall (i,j)\in \IntEnt{1}{d}^2, Q_{ij}(\rho) = \bracket{b_j}{a_i}\bra{a_i}\rho\cket{b_j}.
\end{equation}
One can, but need not, think of the two bases as being eigenbases of two distinct, typically incompatible, observables $A$ and $B$. 
Each KD distribution is a quasiprobability distribution since its values can be nonreal or real and negative. Nevertheless,  it satisfies the following properties: 
\begin{equation}\label{eq:Qmarginals}
\sum_{j=1}^{d} Q_{ij}(\rho)=\langle a_i|\rho|a_i\rangle\in\R, \quad \sum_{i=1}^{d} Q_{ij}(\rho)=\langle b_j|\rho|b_j\rangle\in\R, \quad \sum_{i,j}Q_{ij}(\rho)=\Tr (\rho)=1.
\end{equation}
Thus, the KD distributions' marginals correspond to the probabilities of the Born rule when a measurement is made in one or the other of the two bases. In this sense, a KD distribution is akin to a joint probability distribution for the observables $A$ and $B$. If $\mab := \min_{i,j} |\bracket{a_i}{b_j}| >0$, then the KD distribution is informationally complete: knowledge of $Q(\rho)$ uniquely determines $\rho$ \cite{arvidssonshukur2024properties}. 

It has been shown in various settings such as direct state tomography, quantum metrology,  weak measurements, quantum thermodynamics, quantum scrambling, Leggett-Garg inequalities, and generalised contextuality that nonclassical behaviour of a state $\rho$ is associated with the presence of negative or nonreal values in its KD-distribution. (See \cite{lostaglio2023kirkwood,arvidssonshukur2024properties} for reviews.) This makes it important to obtain a good understanding of what distinguishes such states from the  states for which $Q(\rho)$ is a classical joint probability distribution, such that $Q_{ij}(\rho)\geq 0$ for all $i,j\in\IntEnt{1}{d}$. We shall refer to such states as KD-positive states, following~\cite{langrenez2023characterizing}. Note that in much of the literature the term ``KD-classical state'' is used instead. But since the notions of a classical and nonclassical state can have a variety of meanings, we  use the more specific and neutral term ``KD-positive state''. We denote the convex set of all KD-positive states by $\EcalKDC$. In view of what precedes, it is of interest to identify witnesses and measures of KD nonpositivity. This is the objective of this article.

Writing $\EcalKDCext$ for the extreme points of $\EcalKDC$ and $\EcalKDCpu$ for the pure KD positive states, one has that
\begin{equation}
    \Acal\cup\Bcal\subseteq \EcalKDCpu\subseteq \EcalKDCext,
\end{equation}
where
\begin{equation}\label{eq:acalbcal}
    \Acal=\{ \cket{a_i}\bra{a_i}\mid i\in\IntEnt{1}{d}\} \ \mathrm{and} \   \Bcal=\{ \cket{b_j}\bra{b_j}\mid j\in\IntEnt{1}{d}\}
\end{equation}
are eigenbasis projectors of $A$ and $B$. As outlined in Fig.~\ref{fig:inclusions}, it follows that
\begin{equation}\label{eq:inclusions}
\conv{\Acal\cup\Bcal}\subseteq \conv{\EcalKDCpu}\subseteq \EcalKDC,
\end{equation}
where ``conv'' abbreviates ``convex hull of''. The first inclusion is readily checked from the definition of the KD distribution. 

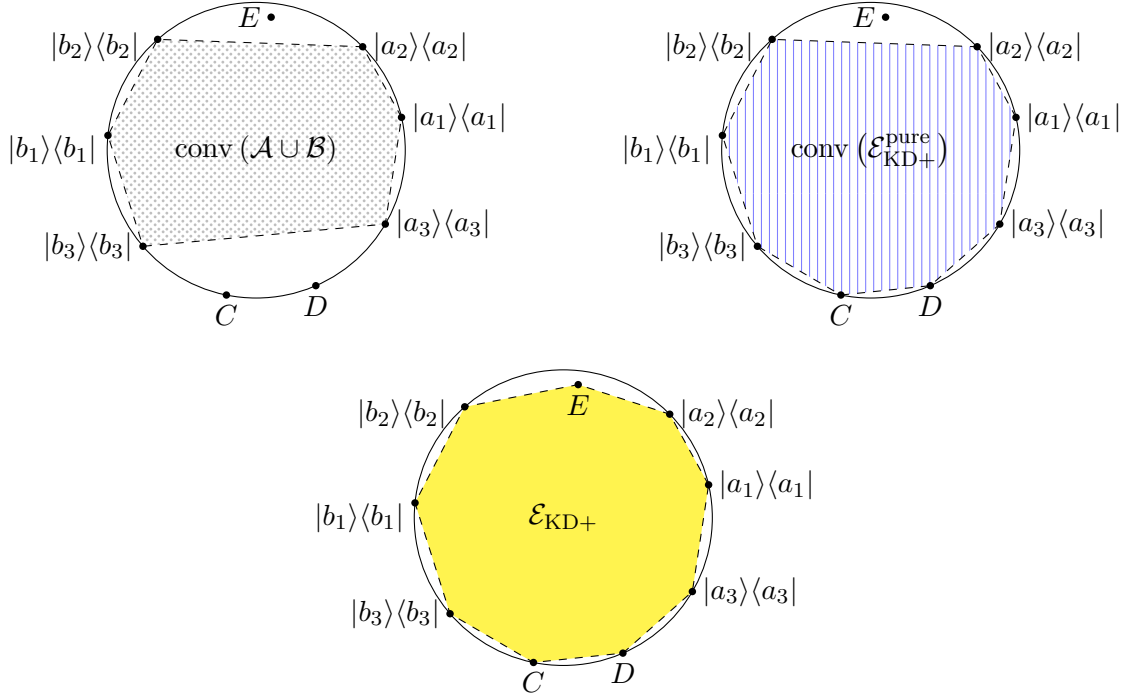
\begin{figure}
\begin{center}
\begin{tikzpicture}[scale=0.98]
	\draw (0,0) circle (2cm);
	\fill[pattern = crosshatch dots, pattern color=gray!50!white](1.95,{sqrt(4-1.95*1.95)}) -- ({sqrt(4-1.4*1.4)},1.4) --({-sqrt(4-1.5*1.5)},1.5)-- ({-sqrt(4-0.2*0.2)},0.2)-- (-{sqrt(4-1.3*1.3)},-1.3)--({sqrt(4-1*1)},-1)--(1.95,{sqrt(4-1.95*1.95)});
	\draw[color=black, dashed](1.95,{sqrt(4-1.95*1.95)}) -- ({sqrt(4-1.4*1.4)},1.4) --({-sqrt(4-1.5*1.5)},1.5)-- ({-sqrt(4-0.2*0.2)},0.2)-- (-{sqrt(4-1.3*1.3)},-1.3)--({sqrt(4-1*1)},-1)--(1.95,{sqrt(4-1.95*1.95)});
	\node at (0,0) {$\convAB$};
        \fill (1.95,{sqrt(4-1.95*1.95)}) circle (0.05); 
	\node[right] at (1.95,{sqrt(4-1.95*1.95)}){$\cket{a_1}\bra{a_1}$};
	\fill ({sqrt(4-1.4*1.4)},1.4) circle (0.05); 
	\node[right] at  ({sqrt(4-1.4*1.4)},1.4){$\cket{a_2}\bra{a_2}$};
	\fill ({sqrt(4-1*1)},-1) circle (0.05); 
	\node[right] at ({sqrt(4-1*1)},-1){$\cket{a_3}\bra{a_3}$};
	\fill ({-sqrt(4-0.2*0.2)},0.2) circle (0.05); 
	\node[left] at (-2,0){$\cket{b_1}\bra{b_1}$};
	\fill ({-sqrt(4-1.5*1.5)},1.5) circle (0.05); 
	\node[left] at ({-sqrt(4-1.4*1.4)},1.4){$\cket{b_2}\bra{b_2}$};
	\fill (-{sqrt(4-1.3*1.3)},-1.3) circle (0.05); 
	\node[left] at (-{sqrt(4-1.3*1.3)},-1.3){$\cket{b_3}\bra{b_3}$};
	\fill (0.2,1.8) circle (0.05); 
	\node[left] at  (0.2,1.8){$E$};
	\fill (-0.4,{-sqrt(4-0.4*0.4)}) circle (0.05); 
	\node[below] at  (-0.4,{-sqrt(4-0.4*0.4)}) {$C$};
 	\fill (0.8,{-sqrt(4-0.8*0.8)}) circle (0.05); 
	\node[below] at  (0.8,{-sqrt(4-0.8*0.8)}) {$D$};
\end{tikzpicture}
\hspace{1cm}
\begin{tikzpicture}[scale=0.98]
	\draw (0,0) circle (2cm);
	\fill[pattern = vertical lines, pattern color=blue!50!white] (1.95,{sqrt(4-1.95*1.95)}) -- ({sqrt(4-1.4*1.4)},1.4)--({-sqrt(4-1.5*1.5)},1.5)-- ({-sqrt(4-0.2*0.2)},0.2)-- (-{sqrt(4-1.3*1.3)},-1.3)--(-0.4,{-sqrt(4-0.4*0.4)})-- (0.8,{-sqrt(4-0.8*0.8)})--({sqrt(4-1*1)},-1)--(1.95,{sqrt(4-1.95*1.95)});
	\draw[color=black, dashed] (1.95,{sqrt(4-1.95*1.95)}) -- ({sqrt(4-1.4*1.4)},1.4)--({-sqrt(4-1.5*1.5)},1.5)-- ({-sqrt(4-0.2*0.2)},0.2)-- (-{sqrt(4-1.3*1.3)},-1.3)--(-0.4,{-sqrt(4-0.4*0.4)})-- (0.8,{-sqrt(4-0.8*0.8)})--({sqrt(4-1*1)},-1)--(1.95,{sqrt(4-1.95*1.95)});
	\node at (0,0) {$\conv{\EcalKDCpu}$};
\fill (1.95,{sqrt(4-1.95*1.95)}) circle (0.05); 
	\node[right] at (1.95,{sqrt(4-1.95*1.95)}){$\cket{a_1}\bra{a_1}$};
	\fill ({sqrt(4-1.4*1.4)},1.4) circle (0.05); 
	\node[right] at  ({sqrt(4-1.4*1.4)},1.4){$\cket{a_2}\bra{a_2}$};
	\fill ({sqrt(4-1*1)},-1) circle (0.05); 
	\node[right] at ({sqrt(4-1*1)},-1){$\cket{a_3}\bra{a_3}$};
	\fill ({-sqrt(4-0.2*0.2)},0.2) circle (0.05); 
	\node[left] at (-2,0){$\cket{b_1}\bra{b_1}$};
	\fill ({-sqrt(4-1.5*1.5)},1.5) circle (0.05); 
	\node[left] at ({-sqrt(4-1.4*1.4)},1.4){$\cket{b_2}\bra{b_2}$};
	\fill (-{sqrt(4-1.3*1.3)},-1.3) circle (0.05); 
	\node[left] at (-{sqrt(4-1.3*1.3)},-1.3){$\cket{b_3}\bra{b_3}$};
	\fill (0.2,1.8) circle (0.05); 
	\node[left] at  (0.2,1.8){$E$};
	\fill (-0.4,{-sqrt(4-0.4*0.4)}) circle (0.05); 
	\node[below] at  (-0.4,{-sqrt(4-0.4*0.4)}) {$C$};
 	\fill (0.8,{-sqrt(4-0.8*0.8)}) circle (0.05); 
	\node[below] at  (0.8,{-sqrt(4-0.8*0.8)}) {$D$};
\end{tikzpicture}
\end{center}
\begin{center}
\begin{tikzpicture}[scale=0.98]
	\draw (0,0) circle (2cm);
	\fill[color=yellow!80!white] (1.95,{sqrt(4-1.95*1.95)}) -- ({sqrt(4-1.4*1.4)},1.4) --(0.2,1.8)--({-sqrt(4-1.5*1.5)},1.5)-- ({-sqrt(4-0.2*0.2)},0.2)-- (-{sqrt(4-1.3*1.3)},-1.3)--(-0.4,{-sqrt(4-0.4*0.4)})-- (0.8,{-sqrt(4-0.8*0.8)})--({sqrt(4-1*1)},-1)--(1.95,{sqrt(4-1.95*1.95)});
	\draw[color=black,dashed](1.95,{sqrt(4-1.95*1.95)}) -- ({sqrt(4-1.4*1.4)},1.4) --(0.2,1.8)--({-sqrt(4-1.5*1.5)},1.5)-- ({-sqrt(4-0.2*0.2)},0.2)-- (-{sqrt(4-1.3*1.3)},-1.3)--(-0.4,{-sqrt(4-0.4*0.4)})-- (0.8,{-sqrt(4-0.8*0.8)})--({sqrt(4-1*1)},-1)--(1.95,{sqrt(4-1.95*1.95)});
	\node at (0,0) {$\EcalKDC$};
        \fill (1.95,{sqrt(4-1.95*1.95)}) circle (0.05); 
	\node[right] at (1.95,{sqrt(4-1.95*1.95)}){$\cket{a_1}\bra{a_1}$};
	\fill ({sqrt(4-1.4*1.4)},1.4) circle (0.05); 
	\node[right] at  ({sqrt(4-1.4*1.4)},1.4){$\cket{a_2}\bra{a_2}$};
	\fill ({sqrt(4-1*1)},-1) circle (0.05); 
	\node[right] at ({sqrt(4-1*1)},-1){$\cket{a_3}\bra{a_3}$};
	\fill ({-sqrt(4-0.2*0.2)},0.2) circle (0.05); 
	\node[left] at (-2,0){$\cket{b_1}\bra{b_1}$};
	\fill ({-sqrt(4-1.5*1.5)},1.5) circle (0.05); 
	\node[left] at ({-sqrt(4-1.4*1.4)},1.4){$\cket{b_2}\bra{b_2}$};
	\fill (-{sqrt(4-1.3*1.3)},-1.3) circle (0.05); 
	\node[left] at (-{sqrt(4-1.3*1.3)},-1.3){$\cket{b_3}\bra{b_3}$};
	\fill (0.2,1.8) circle (0.05); 
	\node[below] at  (0.2,1.8){$E$};
	\fill (-0.4,{-sqrt(4-0.4*0.4)}) circle (0.05); 
	\node[below] at  (-0.4,{-sqrt(4-0.4*0.4)}) {$C$};
 	\fill (0.8,{-sqrt(4-0.8*0.8)}) circle (0.05); 
	\node[below] at  (0.8,{-sqrt(4-0.8*0.8)}) {$D$};
\end{tikzpicture}
    \caption{These three figures provide a schematic representation of the inclusions in Eq.~\eqref{eq:inclusions}, when $d=3$. The disc represents all states, with the circle forming its boundary representing the pure states. The points $C$ and $D$ represent pure states that are extreme states of $\EcalKDC$, but that are not basis states. If at least one such state exists, $\conv{\Acal\cup\Bcal}\subsetneq\conv{\EcalKDCpu}$. The point $E$ represents a mixed extreme  state of $\EcalKDC$. If such a state exists, $\conv{\EcalKDCpu}\subsetneq\EcalKDC$.}
    \label{fig:inclusions}
\end{center}
\end{figure}

As pointed out above, the pure P-positive and Wigner-positive states 
%\Dave{[Is this true, do we know all pure Wigner positive states in all dimensions?]} \Stephan{[I think so. These are the stabilizer states, it seems to me. Christopher, can you check this out in the literature?]} \Dave{[speaking with Ryuji Takagi at this conference in Tokyo, he thinks it is not fully known in even dimension]}
are readily identified. 
%\Stephan{Actually, throughout this introduction, I am thinking of continuous variable Wigner functions. This is quite clear from the link made with the P function. I mentioned the Hudson theorem above. I would just leave things as they are. } \chris{There is "no Wigner function" in even dimensions I would say. And for odd dimensions, this is Gross paper in 2006. The abstract begins with : "We show that, on a Hilbert space of odd dimension, the only pure states to possess a non-negative Wigner function are stabilizer states.".}
The situation for KD distributions is considerably more complex. This is because the identification of the pure KD-positive states depends strongly on the two bases $\Acal$ and $\Bcal$ used to define it, as illustrated in \cite{arvidsson-shukuretal2021, debievre2021, debievre2023a, langrenez2023characterizing, XU2024, Xu22, langrenezetal2024}.  Nevertheless, in~\cite{arvidsson-shukuretal2021, debievre2021, debievre2023a}, general sufficient conditions were formulated for the nonpositivity of the KD distributions of pure states using a suitably adapted notion of ``uncertainty'',  in analogy with the case of the Wigner and P functions. 
One can phrase the indeterminacy  of  a pure state $\cket{\psi}$ with respect to measurements in the two bases $\Acal$ and $  \Bcal$ in terms of a quantity called the support uncertainty $\nab$ [see Eq.~\eqref{eq:nab}]. It was then shown in~\cite{arvidsson-shukuretal2021, debievre2021, debievre2023a} that, for pure states, a high support uncertainty implies KD-nonpositivity.  The precise statement of these results is recalled in Section~\ref{s:suppuncpure}. While this criterion provides a witness of KD-nonpositivity for pure states, it fails for mixed states, as explained in Section~\ref{s:suppuncmixed}. 
In fact, as we shall show,  any finite level of  depolarisation of a pure KD-positive state renders the aforementioned nonpositivity criterion invalid (Lemma~\ref{lem:uncsuppfail}).  Given that all experiments necessarily include some diversion from pure states, a generalisation of the previous results on pure states to mixed states is necessary for both foundational purposes and for practical applications of the KD distribution. 

We show, in Section~\ref{s:suppuncmixed} (Theorem~\ref{thm:Gen1}), that the convex roof of the support uncertainty, that we denote by $\nabcr$, provides a   necessary condition for a state to be a mixture  of pure  KD-positive states: 
\begin{equation}\label{eq:smallsuppunc}
\rho\in\conv{\EcalKDCpu}\Rightarrow \nabcr(\rho)\leq d+1.
\end{equation}
Convex-roof constructions have proven their efficiency as conceptual tools in other similar problems, such as the construction of witnesses and monotones of entanglement (entanglement of formation), of genuine nongaussianity~\cite{Albarelli_2018} and of continuous variable P nonpositivity for mixed states~\cite{yadin_operational_2018, kwon_nonclassicality_2019}. We recall that, if a function $s$ is defined on the extreme points $K_{\textrm{ext}}$ of a convex set $K$, then its convex roof $\scr$ is the largest convex function defined on the full convex set $K$, and that coincides  with $s$ on its extreme points. We recall in Appendix~\ref{s:convexroofs} the basic properties of convex roofs needed in the body of the paper.

A  natural question then arises: Is it true that
\begin{equation}
\rho\in\EcalKDC\Rightarrow \nabcr(\rho)\leq d+1?
\end{equation}
In other words, do all KD-positive states have low support uncertainty, or, equivalently, is $\nabcr$ a witness for nonpositivity? 
We answer this question in the negative by exhibiting, for a spin-$1$ system, a state $\rho\in\EcalKDC$ for which $\nabcr(\rho)>4$.  In other words, even for this completely incompatible system (definition in Section \ref{s:coinc}), not all KD-positive states have minimal support uncertainty.

We further note that, whereas the implication in Eq.~\eqref{eq:smallsuppunc} shows that mixtures of KD-positive pure states have low support uncertainty, the converse implication is not true, as explained in more detail in Section~\ref{s:suppuncpure}. As a result, $\nabcr$ is not a faithful witness of $\conv{\EcalKDCpu}$. 

We show in Section~\ref{s:NCroof} that a faithful witness of $\conv{\EcalKDCpu}$ is provided by the convex roof $\Ncr(\rho)$ of the total KD nonpositivity $\Ncal(\rho)$ \cite{GoHaDr19}, which is
defined as
\begin{equation}\label{eq:N}
\Ncal(\rho)=\sum_{i,j} |Q_{ij}(\rho)|\geq 1.
\end{equation}
The total KD nonpositivity $\Ncal(\rho)$ is  completely analogous to the Wigner negative volume, introduced in~\cite{kenfack_negativity_2004} and often used as a witness of Wigner negativity.
The total KD nonpositivity is well known to be  a faithful witness of $\EcalKDC$ since $\Ncal(\rho)>1$ if and only if $\rho$ is not KD positive and $\Ncal(\rho)=1$ otherwise:
\begin{equation}\label{eq:Nfaithful}
\Ncal(\rho)=1\ \Leftrightarrow \ \rho\in \EcalKDC.
\end{equation}
  The total KD nonpositivity is often thought of as measuring the amount by which the KD distribution of $\rho$ fails to be a genuine probability. This viewpoint explains why several authors think of $\Ncal(\rho)$ as a measure of nonclassicality. As a result, they refer to $\Ncal(\rho)$ as the total KD nonclassicality of $\rho$. As discussed above, we prefer to use the term total KD nonpositivity. 
 
 In Section~\ref{s:NCroof}, we will show that the convex roof $\Ncr(\rho)$ of $\Ncal(\rho)$ is a faithful witness of $\conv{\EcalKDCpu}$. From this, we will infer that,  in fact, for states $\rho$ belonging to $\EcalKDC$ but not to $\conv{\EcalKDCpu}$, $\Ncal(\rho)$ underestimates the nonpositive nature of the state $\rho$ and that its convex roof $\Ncr(\rho)$ is better suited for this purpose.

%%%%%%%%%%%%%%%%%

\section{Preliminaries: Previous results on KD distributions}\label{s:KDpumixpos}
\subsection{The geometry of the KD-positive states}
We sum up here results of~\cite{debievre2021, debievre2023a,langrenez2023characterizing,langrenezetal2024,XU2024,yangetal2023a}, on which this work relies. 
In view of what precedes, the simplest geometry for $\EcalKDC$ is obtained when
\begin{equation}\label{eq:Extegal}
\Acal\cup\Bcal = \EcalKDCpu = \EcalKDCext.
\end{equation}
In that case $\EcalKDC=\convAB$, and thus $\EcalKDC$ is then a polytope of minimal size with $2d$  known vertices. There are then no pure KD-positive states other than the basis states and there are no extreme KD-positive states that are mixed: see Fig.~\ref{fig:inclusions}.
The following result sums up a number of conditions on $\Acal$ and $\Bcal$ under which this situation prevails~\cite{langrenez2023characterizing,langrenezetal2024}.
\begin{Theorem} \label{thm:simplegeom}
Suppose $\mab=\min_{(i,j)\in\IntEnt{1}{d}^2} \left|\bracket{a_i}{b_j}\right|> 0$. Then,  Eq. \eqref{eq:Extegal} is true under any one of the following conditions:
\begin{enumerate}
	\item If $d=2$;
	\item If $d$ is a prime number and the transition matrix $U=(\bracket{a_i}{b_j})_{(i,j)\in\IntEnt{1}{d}^2}$ is equal to the discrete Fourier transform (DFT) matrix;
	\item In any dimension $d\geq 2$, for a dense open set of transition matrices  $U=(\bracket{a_i}{b_j})_{(i,j)\in\IntEnt{1}{d}^2}$ that has probability 1 among the unitary matrices. 
\end{enumerate}
\end{Theorem} 
In specific cases, the structure of $\EcalKDC$ as a convex set can be considerably more complicated  (see Fig. \ref{fig:inclusions}.)  An example  is given in dimension $3$ in \cite{langrenez2023characterizing}, where the inclusions in Eq.~\eqref{eq:inclusions} are all strict:
\begin{equation}\label{eq:sharpincl}
\Acal\cup\Bcal\subsetneq \EcalKDCpu\subsetneq \EcalKDCext.
\end{equation}
This example is constructed in a spin $1$ system for appropriate choices of the spin component. Then, there exist pure KD-positive states that are not basis states and there exist extreme KD-positive states that are mixed. Such states cannot be written as convex combinations of KD-positive pure states. Moreover,  examples are provided in~\cite{langrenez2023characterizing} where 
\[
\Acal\cup\Bcal = \EcalKDCpu\subsetneq \EcalKDCext
\]
for arbitrarily high Hilbert space dimensions. In these cases, there do exist mixed KD-positive extreme states but  the only pure KD-positive states are the basis projectors. As a rule, given $\Acal$ and $\Bcal$, it is difficult to explicitly determine $\EcalKDCpu$ and more difficult still to determine $\EcalKDCext$ \cite{langrenezetal2024}. 
More information about $\EcalKDCpu$ for the case where the transition matrix is a Hadamard matrix (for all $(i,j)\in\IntEnt{1}{d}^2, \left|\bracket{a_i}{b_j}\right|^2 = \frac{1}{d} = \mab$) so that the two bases are mutually unbiased, can be found in~\cite{debievre2021,debievre2023a,XU2024}. Nevertheless, the third part of Theorem~\ref{thm:simplegeom} ensures that, with probability one, such complications and subtleties do not occur. 

\subsection{Complete incompatibility}\label{s:coinc}
%Because of the potentially intricate structure of $\EcalKDC$ as described above, we turn in the next sections to the identification of witnesses capable of distinguishing $\conv{\EcalKDCpu}$ from $\EcalKDC$.

The condition $\mab>0$, in Theorem~\ref{thm:simplegeom}, can be interpreted as a weak form of incompatibility for the measurements associated with $\Acal$ and $\Bcal$. It implies that, if a measurement in the $\Acal$ basis is made first, then a subsequent measurement in the $\Bcal$ basis can, with nonvanishing probability, yield any possible outcome. A significantly stronger notion of \emph{complete incompatibility} of two bases is introduced in~\cite{debievre2021}. Complete incompatibility occurs in scenarios where, after a first (coarse- or fine-grained) measurement in the $\Acal$ basis is made, a subsequent measurement in the $\Bcal$ basis is sure to perturb the outcome of the first measurement. In mathematical terms, complete incompatibility is equivalent to the requirement that all minors of the transition matrix $U$ between the two bases are nonvanishing; we will say that such a unitary matrix, and its underlying bases, are completely incompatible. This  implies a  strong form of incompatibility: for all $S,T\subseteq \llbracket 1,d\rrbracket$, with $1\leq |S|, |T|<d$ here $|X|$ denotes the cardinality of $X$), one has that $[\PiAcal(S),\PiBcal(T)]\not=0$,
where the projectors $\PiAcal(S), \PiBcal(T)$ are defined as
$$
\PiAcal(S)=\sum_{i\in S} |a_i\rangle\langle a_i|,\quad \PiBcal(T)=\sum_{j\in T} |b_j\rangle\langle b_j|. 
$$
Complete incompatibility has been introduced and studied in \cite{debievre2021, debievre2023a,yangetal2023a}. When the Hilbert space dimension is  $2$ or $3$, two bases are completely incompatible if and only if $\mab > 0$. This is not true in higher dimensions. Nevertheless, one then has the following result.
\begin{Theorem}
    In any finite dimension, there exists a dense, open, probability $1$ set of transition matrices that are completely incompatible.
\end{Theorem}
The ``open and dense'' part of this result was shown in~\cite{debievre2023a}. The ``probability one'' part is an immediate result of Proposition 2.4 in~\cite{langrenezetal2024}, which states that the zero set of a polynomial in the matrix elements of $U$ form a set of vanishing Haar measure in the unitary group. Since complete incompatibility requires all minors of $U$ to be nonvanishing, the result then follows immediately. 

Let us finally point out that two bases whose transition matrix coincides with the discrete Fourier transform (DFT) are completely incompatible if and only if the Hilbert space dimension is prime~\cite{debievre2023a}.

\subsection{Support uncertainty as a nonpositivity witness: pure states}\label{s:suppuncpure}
 Given a pure state $|\psi\rangle\in\Hcal$, one defines $\na(\psi)$ [respectively $\nb(\psi)$] to be the number of nonvanishing components of $|\psi\rangle$ on $\Acal$ [respectively $\Bcal$]:
\begin{eqnarray}\label{eq:naetnb}
\na(\psi)=\left|\{i\in\llbracket 1,d\rrbracket \mid \langle a_i|\psi\rangle\not=0\}\right|,\\ \nb(\psi)=\left|\{j\in\llbracket 1,d\rrbracket \mid \langle b_j|\psi\rangle\not=0\}\right|.
\end{eqnarray}
It was proven in~\cite{debievre2021} that, if $\mab>0$, then
\begin{equation}\label{eq:nab}
\nab(\psi):=\na(\psi)+\nb(\psi)>d+1
\end{equation}
implies that $|\psi\rangle$ is not KD positive. Equivalently, one has
\begin{equation}\label{eq:KDplusnab}
    |\psi\rangle\in \EcalKDCpu\ \Rightarrow\ \nab(\psi)\leq d+1.
\end{equation}
The \emph{support uncertainty} $\nab(\psi)$ therefore provides a KD nonpositivity witness for pure states. However, $\nab(\psi)$ is not a faithful witness:  there may exist states $|\psi\rangle$ for which $\nab(\psi)\leq d+1$ that are not KD positive. Examples of bases $\Acal$ and $\Bcal$ for which this occurs are provided in~\cite{debievre2021, debievre2023a}.  The link between positivity and uncertainty is therefore not as tight for the KD distribution as for the Wigner distribution and the P function.  Indeed, for the Wigner distribution,  $\det\gamma$ [Eq. \eqref{eq:covmatrix}] provides a faithful Wigner nonpositivity witness for pure states and, for the P function, $\Delta X^2+\Delta P^2$ provides a faithful P nonpositivity witness.

Nevertheless, when the bases $\Acal$ and $\Bcal$ are completely incompatible, more can be said. First, in this case, for all pure states $|\psi\rangle$
\begin{equation}\label{eq:nabcomplinc}
\nab(\psi)\geq d+1.
\end{equation}
Hence, for all pure KD-positive states, $\nab(\psi)= d+1$~\cite{debievre2021}. This is somewhat analogous to what happens with the Wigner or P functions for systems with two conjugate variables,  which are positive for pure states if and only if the state is has minimal uncertainty in the sense explained above. The difference is that, with the KD distribution, even if $\Acal$ and $\Bcal$ are completely incompatible, it is still possible for states with minimal support uncertainty to  be KD nonpositive. In other words, $\nab$ is not a faithful KD nonpositivity witness, even for completely incompatible bases.

\section{Support uncertainty as a nonpositivity witness: mixed states}\label{s:suppuncmixed}
Our main result in this section is Theorem~\ref{thm:Gen1}, that extends to mixed states the relation between KD positivity and support uncertainty shown in~\cite{debievre2021, debievre2023a} for pure states and recalled in the previous section. Before doing so, we point out that a naive extension of the notion of support uncertainty to mixed states $\rho$ is not satisfactory. Let us define
\[
\na(\rho) = |\{i\in\IntEnt{1}{d} \mid \bra{a_i}\rho\cket{a_i}\not=0\}| ,\quad \nb(\rho)=|\{j\in\IntEnt{1}{d} \mid | \bra{b_j}\rho\cket{b_j}\not=0\}|
\]
and 
$$
\nab(\rho):=\na(\rho)+\nb(\rho)\leq 2d.
$$
We  observe that, if $\mab > 0$, and $\rho = \frac{1}{d}\id{d}$ , then $\nab(\rho) = 2d>d+1$. Since $\rho = \frac{1}{d}\id{d}$ belongs to  $\conv{\EcalKDCpu}$, it becomes  clear that, if we define $\nab(\rho)$ to be the support uncertainty of $\rho$, then Eq.~\eqref{eq:KDplusnab} does not extend to mixed KD-positive states. In fact, the following lemma generalizes this observation.
\begin{Lemma}\label{lem:uncsuppfail}
Let $\mab>0$. Let  $\rho\in \convAB$. Suppose $\rho$ is written in the form
\begin{equation}\label{eq:convdecomp}
\rho=\sum_{i=1}^d\lambda_i|a_i\rangle\langle a_i|+\sum_{j=1}^d\mu_j|b_j\rangle\langle b_j|,
\end{equation}
with $0\leq \lambda_i,\mu_j\leq 1$.  
Then the following statements are true:\\
(i) If there exist $i,j$ for which $\lambda_i\not=0\not=\mu_j$, then $\nab(\rho)=2d$;\\
(ii) If $\mu_j=0$  for all $j$, then 
$$
\nab(\rho)=d+\left|\{i| \lambda_i\not=0 \}\right|;
$$
(iii) If $\lambda_i=0$ for all $i$, then 
$$
\nab(\rho)=d+\left|\{j| \mu_j\not=0 \}\right|.
$$
In particular, if $\rho$ belongs to the interior of the convex polytope $\convAB$, then $\nab(\rho)=2d$. 
\end{Lemma}

\begin{proof}
(i) Suppose, without loss of generality, that $\lambda_1\not=0\not=\mu_1$. Then, for any $i,j$
\begin{eqnarray*}
\langle a_i|\rho|a_i\rangle&\geq& \lambda_1|\langle a_i|a_1\rangle|^2 +\mu_1|\langle a_i|b_1\rangle|^2>0,\\
\langle b_j|\rho|b_j\rangle&\geq& \lambda_1|\langle b_j|a_1\rangle|^2 +\mu_1|\langle b_j|b_1\rangle|^2>0.
\end{eqnarray*}
Statements (ii)  \& (iii)  follow from a direct computation using that $\mab>0$. 
To prove the last statement, we note that it was proven in~\cite{langrenez2023characterizing} (Appendix~A) that $\rho$ is in the interior of $\convAB$ if and only if  it can be written in the form Eq.~\eqref{eq:convdecomp} with 
\begin{equation}
\lambda_1\lambda_2\dots\lambda_d\not=0\ \textrm{or}\ \mu_1\mu_2\dots\mu_d\not=0.
\end{equation}
If $\lambda_1\lambda_2\dots\lambda_d\not=0$ and  $\mu_j\not=0$ for at least one value of $j$, then the conclusion follows from~(i). If, on the other hand all $\mu_j=0$, then it follows from~(ii). 
\end{proof}

The Lemma implies that the minimal value of $\nab$ on the polytope $\convAB\subset \conv{\EcalKDCpu}$ is $d+1$ and that this value is reached only on its vertices. In addition, the global maximum of $\nab$ is reached on each point of the interior of $\convAB$. On the facets of $\convAB$ one finds all intermediate integer values. Since all states in $\convAB$ are KD positive, this amply demonstrates that there is no link between a small value of $\nab(\rho) $ and KD positivity of $\rho$.

To remedy this situation, we provide below a proper definition of the ``support uncertainty'' of a mixed state using a convex roof construction. This will allow us to show that the results of~\cite{debievre2021, debievre2023a}  linking KD positivity and support uncertainty extend naturally to mixed states as well. 
Our work relies on some technical properties of convex roofs that are provided in Appendix \ref{s:convexroofs}.

For any state $\rho$, we define $D_{\rho}$ to be the set of all possible convex decompositions of $\rho$ on pure states, \textit{i.e.}
$$
D_{\rho}= \bigcup_{n\in\N^{*}}\left\{(\lambda_{i},\cket{\varphi_i})_{i\in \IntEnt{1}{n}}\in([0,1]\times \H_{1})^{n}\left| \sum_{i\in \IntEnt{1}{n}} \lambda_{i}\cket{\varphi_i}\bra{\varphi_i} = \rho \right.\right\} , 
$$
with $\H_{1} = \left\{\cket{\psi},\bracket{\psi}{\psi} = 1\right\}$ denoting the unit sphere in $\H$.

\begin{defini}\label{def:Ext1}
Let $\rho$ be a density matrix. Then we define the support uncertainty $\nabcr(\rho)$  of $\rho$ as
%and
\[
\nabcr(\rho) = \inf_{(\lambda_{i},\cket{\varphi_i})_{i\in \IntEnt{1}{n}}\in D_{\rho}} \sum_{i=1}^{n} \lambda_{i} \nab(\varphi_{i}).\]  
\end{defini}

By Lemma \ref{lem:CR2}, the infimum is in fact a minimum. In other words, there exists $n\in\N^{*}$ and $(\lambda_{i},\cket{\varphi_i})_{i\in \IntEnt{1}{n}}$ such that $\rho = \sum_{i=1}^{n} \lambda_{i}\cket{\varphi_{i}}\bra{\varphi_{i}}$ and 

\[\nabcr(\rho) = \sum_{i=1}^{n} \lambda_{i} n_{\mcl{A},\mcl{B}}(\varphi_{i}).
\]
We then have the following result.
 
\begin{Theorem}\label{thm:Gen1}
    Suppose that $\mab >0$. If  $\rho\in\conv{\EcalKDCpu}$,  then $\nabcr(\rho) \leqslant d+1$. If $\Acal$ and $\Bcal$ are completely incompatible, then $\rho\in\conv{\EcalKDCpu}$ implies that $\nabcr(\rho)=d+1$. 
\end{Theorem}
One way to interpret this result is to say that $\nabcr$ is a witness for $\conv{\EcalKDCpu}$: if $\nabcr(\rho)>d+1$, then $\rho\not\in\conv{\EcalKDCpu}$. 
In addition, when the bases $\Acal$ and $\Bcal$ are completely incompatible, $\nabcr$ reaches its minimum ($d+1$), on the set $\conv{\EcalKDCpu}$. In other words, all mixtures of pure KD-positive states then have minimal support uncertainty.  
\begin{proof}
As $\mab > 0$, if $\cket{\psi}\in\EcalKDCpu$, then, by Eq.~\eqref{eq:KDplusnab}, $\nab(\psi)\leqslant d+1$. Now, if $\rho \in \conv{\EcalKDCpu}$, there exists $n\in\N^{*}$ such that
\[
\rho = \sum_{i=1}^{n} \lambda_i \cket{\psi_i}\bra{\psi_i}
\]
with $\sum_{i=1}^{n} \lambda_i=1, \lambda_i\in[0,1]$ for all $i\in\IntEnt{1}{n}$,  and $\cket{\psi_i}\in\EcalKDCpu$ for all $i\in\IntEnt{1}{n}$. Thus, by definition, 
\[
\nabcr(\rho) \leqslant \sum_{i=1}^{n}\lambda_i\nab(\psi_i)\leqslant \sum_{i=1}^{n}\lambda_i (d+1) \leqslant d+1.
\]
This concludes the first part of the proof. Now, if $\Acal$ and $\Bcal$ are completely incompatible, it follows from Eq.~\eqref{eq:nabcomplinc} that, for all pure $\cket{\psi}$, $\nab(\psi) \geqslant d+1$. Thus, if\[
\rho = \sum_{i=1}^{n} \lambda_i \cket{\psi_i}\bra{\psi_i}
\]
with $n\in\N^{*}, \sum_{i=1}^{n} \lambda_i=1, \lambda_i\in[0,1]$ for all $i\in\IntEnt{1}{n}$,  and $\cket{\psi_i}\in\H_1$ for all $i\in\IntEnt{1}{n}$, it follows that:
\[
\sum_{i=1}^{n}\lambda_i\nab(\psi_i) \geqslant \sum_{i=1}^{n}\lambda_i(d+1) = d+1.
\]
As the right hand side is independent of the decomposition, we conclude that:
\[
\nabcr(\rho) \geqslant d+1.
\]
This inequality, together with the first part of the proof, implies that, if $\rho\in\conv{\EcalKDCpu}$, then $\nabcr(\rho) = d+1$ when $\Acal$ and $\Bcal$ are completely incompatible.
\end{proof}
In view of this result, the question arises naturally whether $\nabcr$ is a witness for all KD-positive states, and not only for $\conv{\EcalKDCpu}$. In other words, the question is whether the implication 
\begin{equation}\label{eq:question2}
\rho\in\EcalKDC\ \Rightarrow \ \{\nabcr(\rho)\leq d+1\} 
\end{equation}
is true or not for any choice of $\Acal$ and $\Bcal$.

To answer this question, we first recall that, according to the third part of Theorem~\ref{thm:simplegeom},  $\EcalKDC=\convAB$ holds in any dimension for an open dense set of unitary transition matrices $U_{ij}=\langle a_i|b_j\rangle$ of full Haar measure in the unitary group. In those cases, $\EcalKDCpu=\Acal\cup \Bcal$, so that the implication in Eq.~\eqref{eq:question2} is satisfied as as consequence of Theorem~\ref{thm:Gen1} and therefore $\nabcr$ is a witness for $\EcalKDC$ with probability $1$ in any dimension $d$.
The question is therefore if the implication in Eq.~\eqref{eq:question2} holds for all choices of $\Acal$ and $\Bcal$. We show that this is not the case by further examining the example of a spin-$1$ system in dimension $3$ for which it was shown in~\cite{langrenez2023characterizing} that Eq.~\eqref{eq:sharpincl} holds for an appropriate choice of two spin observables. We explicitly exhibit, in Appendix~\ref{App:spin1}, a $d=3$ state $\rho$ in $\EcalKDC$  for which $\nabcr(\rho)>4$.

\section{The total KD nonpositivity roof}\label{s:NCroof}
We have seen, in Eq.~\eqref{eq:Nfaithful}, that the total KD nonpositivity $\Ncal$ is a faithful KD nonpositivity witness. 
We now briefly show, using the results of Appendix~\ref{s:convexroofs}, that the convex roof $\Ncr$ of the nonpositivity $\Ncal$ is a faithful witness of $\conv{\EcalKDCpu}$.

We define $\Ncr$ as the convex roof of $\Ncal$:
 \begin{equation}
\Ncr(\rho) = \inf_{(\lambda_{i},\varphi_i)_{i\in \IntEnt{1}{n}}\in D_{\rho}} \sum_{i=1}^{n} \lambda_{i} \mcl{N}(\varphi_{i}).
\end{equation}
Since $\Ncal$ is convex, one has for any $\rho$,
$$
\Ncal(\rho)\leq \Ncr(\rho).
$$
The total KD nonpositivity $\Ncal$ is  a continuous function that has $1$ as its minimal value. For pure states, $\Ncal(\rho)=1$ if and only if $\rho \in  \EcalKDCpu$. It thus follows from Proposition~\ref{prop:CR1}  that 
\begin{equation}\label{eq:NCroof1}
 \Ncr(\rho)=1 \Leftrightarrow \rho\in\conv{\EcalKDCpu}.
\end{equation}
In other words, whereas $\Ncal$ is a faithful witness for $\EcalKDC$, $\Ncr$ is a faithful witness for $\conv{\EcalKDCpu}$. 
In particular, 
\begin{equation}
   \rho\in\EcalKDC\setminus \conv{\EcalKDCpu}\Rightarrow \Ncal(\rho)=1, \Ncr(\rho)>1. 
\end{equation}
That is, if $\rho$ is KD positive but cannot be written as a convex combination of pure KD-positive states, then the convex roof of the total KD nonpositivity is strictly greater than one. Examples of such states are given in Appendix~\ref{App:spin1}.

Our results invite the following comparison with entanglement theory. Consider a bi-partite state $\rho^{\textrm{A,B}}$. Its von Neumann entropy is 
\begin{equation}
    S\left(\rho^{\textrm{A,B}}\right)=-\Tr\rho^{\textrm{A}} \ln\rho^{\textrm{A}} =-\Tr\rho^{\textrm{B}} \ln\rho^{\textrm{B}}  ,
\end{equation}
where  $\rho^{\textrm{A}} = \Tr^{\textrm{B}} \left( \rho^{\textrm{A,B}} \right)$ is the partial trace of $\rho^{\textrm{A,B}}$ over the $\textrm{B}$-subsystem, etc. $S\left(\rho^{\textrm{A,B}}\right)$ is a concave function of $\rho^{\textrm{A,B}}$. If $\rho^{\textrm{A,B}}$ is pure, $S\left(\rho^{\textrm{A,B}}\right)$ quantifies its entanglement. However, if $\rho^{\textrm{A,B}}$ is mixed, $S(\rho^{\textrm{A,B}})$ can be large although there is no entanglement.

The entanglement of formation $E_{\textrm f}(\rho^{\textrm{A,B}})$ was constructed to quantify entanglement of  bi-partite mixed states. $E_{\textrm f}(\rho^{\textrm{A,B}})$ is defined as the convex roof, over pure states, of the von Neumann entropy:
\begin{equation}
    E_{\textrm f}(\rho^{\textrm{A,B}}) := 
    \inf_{(\lambda_{i},\varphi_i)_{i\in \IntEnt{1}{n}}\in D_{\rho^{\textrm{A,B}}}} \sum_{i=1}^{n}  \lambda_{i} S\left( \varphi_i^{\textrm{A,B}}\right) .
\end{equation}
The concavity of the von Neumann entropy implies that 
\begin{equation}
    S\left(\rho^{\textrm{A,B}}\right) \geqslant E_{\textrm f}\left(\rho^{\textrm{A,B}}\right):=S^{\textsc{cr}}\left(\rho^{\textrm{A,B}}\right).
\end{equation}
For example, if $\rho^{\textrm{A,B}} $ is the maximally mixed state, then its von Neumann entropy equals $\ln d>0$, while its entanglement of formation vanishes. In this sense, the von Neumann entropy overestimates the entanglement of the state. 

The situation is the converse for the total KD nonpositivity $\Ncal(\rho)$, which is a convex function of its argument. Its convex roof $\Ncr(\rho)$ satisfies
\begin{equation}
    \Ncal(\rho)\leq \Ncr(\rho).
\end{equation}
The total nonpositivity $\Ncal(\rho)$ can therefore underestimate $\Ncr(\rho)$. 
In fact, any state $\rho\in\EcalKDC\setminus\conv{\EcalKDCpu}$ satisfies
$$
1=\Ncal(\rho)< \Ncal^{\textsc{cr}}(\rho).
$$
 Note that such states cannot be written as convex mixtures of pure KD-positive states: any such mixture must contain at least one pure state that is not KD positive. (See Fig.~\ref{fig:inclusions}; examples are given in Ref.~\cite{langrenez2023characterizing}.) This feature of the state $\rho$ is not picked up by its total nonpositivity $\Ncal(\rho)$, but it is by its convex roof $\Ncr(\rho)$.

%%%%%%%%%%%%%%%%%%%%%

\section{Conclusion and discussion}\label{s:concdisc}
Convex roofs are notoriously hard to compute and therefore criticized for being of limited practical value. Nevertheless, our results illustrate again their power as conceptual tools in quantum theory.  We showed in particular that the convex roof  of a quantum state's support uncertainty $\nabcr$ provides a witness---be it a nonfaithful one---for the complement of the set $\conv{\EcalKDCpu}$. Thus,  $\nabcr$ links KD-positivity to an uncertainty principle: mixtures of pure KD-positive states have low support uncertainty.

The convex roof $\Ncr$ of the total KD nonpositivity $\Ncal$, on the other hand, is a faithful witness of the set $\conv{\EcalKDCpu}$. Since $\Ncal$ itself is a faithful witness of $\EcalKDC$ (the set of \textit{all} KD-positive states) these two witnesses together  allow one to distinguish $\conv{\EcalKDCpu}$ from the potentially larger set $\EcalKDC$. In particular, for a  KD-positive state $\rho$  that is not a mixture of KD-positive pure states, $\Ncal(\rho)=1$, while $\Ncr(\rho)>1$. 
These results provide a first insight in the physical interpretation of such states. Note that, in a resource theory of KD-nonpositivity, they would not be considered ``free'' states: only states in $\conv{\EcalKDCpu}$ would be. 

We point out that the problem of distinguishing $\EcalKDC$ from $\conv{\EcalKDCpu}$ has an analog in the context of the Wigner function. There, the set of states obtained as convex
combinations of pure Gaussian states is the analog of $\conv{\EcalKDCpu}$, whereas the larger set of all Wigner-positive states is the analog of $\EcalKDC$.  The identification of these sets through witnesses or
monotones has attracted considerable attention in the last decade in the context of degaussification~\cite{Albarelli_2018, hertzdebievre2023,Genoni13,Takagi18}. This is because, while convex combinations of Gaussian states are not Gaussian, they are obtained
by a mixing operation which is considered classical or, in the context of resource theories, “free”. We are hopeful that the tools to distinguish $\EcalKDC$ from $\conv{\EcalKDCpu}$,  presented in this article, will aid the construction of analogous tools for the Wigner function.

\medskip

\noindent \textit{Acknowledgements: } This work was supported in part by the Agence Nationale de la Recherche under grant ANR-11-LABX-0007-01 (Labex CEMPI), by the Nord-Pas de Calais Regional Council and the European Regional Development Fund through the Contrat de Projets \'Etat-R\'egion (CPER), and by  the CNRS through the MITI interdisciplinary programs. We thank Girton College, Cambridge, for support of this work. D.R.M. Arvidsson-Shukur thanks Nicole Yunger Halpern for useful discussions.

\newpage

\appendix 
\appendixpage
\renewcommand{\thesection}{\Alph{section}}
\renewcommand{\thesubsection}{\Alph{section}.\arabic{subsection}}

\section{Convex roofs}\label{s:convexroofs}
Our goal in this appendix is to show Lemma~\ref{lem:CR2}, which is a result on  convex roofs, needed in the main text. More specifically, it locates the set where a convex roof reaches its minimum.

First, we collect some known facts about convex roofs, some of which are also needed in the main part of the paper. The discussion is self-contained but we refer to~\cite{uhlmann2010}, that we follow, for more details.
Let $K$ be a convex and compact subset  of $\R^{q}$. Let $s: K_{\mathrm{ext}} \mapsto \R$ where $K_{\mathrm{ext}}$ designates the set of extreme points of $K$. We suppose throughout that $s$ is bounded below. %$\exists m\in\R, \forall x\in K_{ext}, f(x) \geqslant m$.
In what follows, we will say
$S$ is a convex extension of $s$  on $K$ if $S$ is convex and $S=s$ on $K_{\mathrm{ext}}$.
The convex roof $\scr$ of $s$ is defined to be the biggest convex extension of $s$. It is given by  
$$
 \forall x\in K, \ \scr(x) = \inf_{(\lambda_{i},x_i)_{i\in \IntEnt{1}{n}}\in D_{x}} \sum_{i=1}^{n} \lambda_{i} s(x_{i}) 
 $$
where 
$$
D_{x}= \bigcup_{n\in\N^{*}}\left\{(\lambda_{i},x_i)_{i\in \IntEnt{1}{n}}\in([0,1]\times K_{\mathrm{ext}})^{n}\left| \sum_{i\in \IntEnt{1}{n}} \lambda_{i}x_i = x, \sum_{i\in \IntEnt{1}{n}}\lambda_i=1\right.\right\}.
$$
Note that, as a result of the Krein-Milman theorem, recalled below, $D_x$ is not empty, so that $\scr$ is well defined. 
\begin{Theorem}[Krein-Milman]
In a finite-dimensional vector space, a convex compact set is the convex hull of its extreme points: $K=\conv{\Kext}$.
\end{Theorem}
Note  that $K_{\mathrm{ext}}$ may not be compact, even if $K$ is. On the other hand, the convex hull of a compact set is always compact, since we are in finite dimension.
We will need the following result, that can be found in \cite{uhlmann2010}.
\begin{Lemma}\label{lem:CR3}
Suppose $K_{\mathrm{ext}}$ is compact and that $s:K_{\mathrm{ext}}\to\R$ is continuous, then for each $x$ in $K$, there exists an optimal decomposition, meaning that there exists $(\lambda_{i},x_i)_{i\in \IntEnt{1}{n}}\in D_x$ such that 
\begin{equation}\label{eq:CR3}
\scr(x) = \sum_{i=1}^{n} \lambda_{i}s(x_i).
\end{equation}
\end{Lemma}

We now formulate the problem we wish to address. It is clear from the above that, in general, 
$$
\scr(x)\geq \inf_{x\in K_{\mathrm{ext}}} s(x)
$$ 
for all $x\in K$. Indeed, let $(\lambda_{i},x_{i})_{i\in\IntEnt{1}{n}}$ be a convex decomposition of $x$. Then
$$
 \sum_{i=1}^{n} \lambda_{i}s(x_{i}) \geqslant  \sum_{i=1}^{n} \lambda_{i} \inf_{x\in K_{\mathrm{ext}}} s(x) \geqslant \inf_{x\in K_{\mathrm{ext}}} s(x).
 $$
Hence
\begin{equation}\label{eq:CR1}
\scr(x) = \inf_{(\lambda_{i},x_i)_{i\in \IntEnt{1}{n}}\in D_{x}} \sum_{i\in \IntEnt{1}{n}} \lambda_{i} s(x_{i}) \geqslant \inf_{x\in K_{\mathrm{ext}}} s(x) .
\end{equation}

In what follows, we suppose $s$ has a minimum $m$ on $K_{\mathrm{ext}}$, which is therefore also a minimum for $\scr$.  We write 
$$
\mcl{M}_s=s^{-1}(m)\subset K_{\mathrm{ext}}
$$ 
for the set where this minimum is reached. We wish to locate  the set 
$$
\mcl{M}_{\scr}=\left(\scr\right)^{-1}(m)\subset K
$$
where $\scr$ is minimal. 
\begin{Prop}\label{prop:CR1}
Let $\scr$ be the convex roof of $s : K_{ext}\mapsto \R$. Suppose $s$ has a minimum $m$  on $K_{\mathrm{ext}}$.
Then $m$ is the minimum of $\scr$ and
$$ 
\mathrm{conv}(\mcl{M}_{s}) \subset \mcl{M}_{\scr}.
$$
If in addition $\Kext$ is compact and $s$ continuous, then
$$
\mathrm{conv}(\mcl{M}_{s}) = \mcl{M}_{\scr}.
$$
\end{Prop}

\begin{proof}
From Eq.~\eqref{eq:CR1} we know $\scr(x)\geq m$ for all $x\in K$ and since $\scr=s$ on  $K_{ext}$, it follows that the minimum of $\scr$ is $m$ and that $\mcl{M}_s\subset\mcl{M}_{\scr}$. Now, let $x\in\text{conv}(\mcl{M}_{s})$. Then $x=\sum_{i=1}^{n} \lambda_{i}x_i$ for some $(\lambda_{i},x_i)_{i\in\IntEnt{1}{n}}\in ([0,1]\times\mcl{M}_{s})^{n}$ . Then, 
\[
\scr(x) \leqslant \sum_{i=1}^{n} \lambda_{i}s(x_i) = \sum_{i=1}^{n} \lambda_{i} m =m.
\]
Together with \eqref{eq:CR1} this implies $\scr(x)=m$ and consequently $\conv{\mcl{M}_{s}}\subset \mcl{M}_{\scr}$.

We now turn to the second part of the proof. Suppose now that $x\in \mcl{M}_{\scr}$ so that $\scr(x)=m$. Lemma~\ref{lem:CR3} guarantees that one can write $\scr(x) = \sum_{j=1}^{n} \mu_{j} s(x_{j})$, with $x_j\in \Kext$. We wish to prove $s(x_j)=m$ for all $j$. Suppose on the contrary that $s(x_{1}) > m$; then 
\[
\scr(x) = \mu_{1} s(x_{1})+ \sum_{i=2}^{n} \mu_{j} s(x_{j}) > \mu_{1}m + \sum_{i=2}^{n} \mu_{j} s(x_{j}) \geqslant \sum_{i=1}^{n} \mu_{j} m =m,
\]
which is a contradiction.  We conclude that for all $j\in\IntEnt{1}{n}, \  s(x_{j})=m$ and thus $ \mcl{M}_{\scr}\subset  \text{conv}(\mcl{M}_{s})$ which proves the reverse inclusion.

\end{proof}
One may ask what happens if $s$ is not continuous. Does one have $ \overline{\mathrm{conv}}(\mcl{M}_{s}) \subset \mcl{M}_{\scr}$?  The answer is ``no''.  To see this, it is enough to look at an example where $K$ is the disc of radius $1$, centered at the origin of  the plane $\R^2$. Hence $K_{\textrm{ext}}$ is the unit circle. We take $s$ to be 
\begin{equation}\label{eq:sexample}
s:(x,y)\in K_{\textrm{ext}}  \rightarrow \chi_{[0,+\infty[}(y)\in\R :
\end{equation}
$s$ vanishes on the open lower half unit circle and equals $1$ on the closed upper half unit circle. In that case $\mathcal M_s$ is the open lower half unit circle and consequently $\conv{\mathcal M_s}$ is the open lower half disc. On the other hand, through a simple limiting procedure, one can show that $\mathcal M_{\scr}$ is the closed lower half disc with $(-1,0)$ and $(1,0)$ removed. In particular, $\mcl{M}_{\scr}$ is not closed and $\overline{\mathrm{conv}}(\mcl{M}_{s}) \not\subset \mcl{M}_{\scr}$. In Lemma~\ref{lem:CR2} below, we show that this is the generic situation in the presence of a gap between the minimal value of $s$ and the other values of $s$.

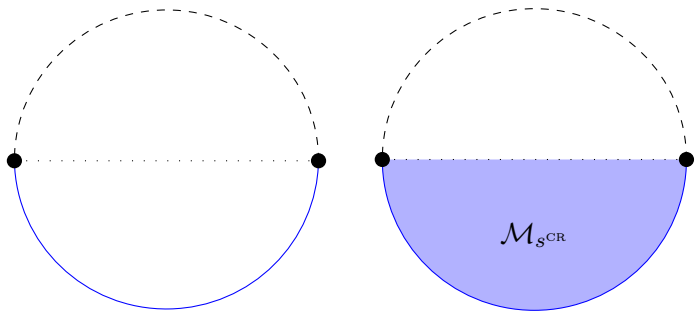
\begin{figure}[!h]
    \begin{center}
    \begin{tikzpicture}[scale=1]
        \draw[color=blue] (-2,0) arc (181:359:2cm);
        \draw[dashed] (2,0) arc (0:180:2cm);
        \draw[loosely dotted] (-2,0)--(2,0);
        \fill (-2,0) circle (0.1);
        \fill (2,0) circle (0.1);

    \end{tikzpicture}
    \hspace{1em}
    \begin{tikzpicture}[scale=1]
        \fill[color=blue!30!white] (-2,0) arc (180:360:2cm);
        \draw[dashed] (2,0) arc (0:180:2cm);
        \draw[loosely dotted] (-2,0)--(2,0);
        \node at (0,-1){$\mcl{M}_{\scr}$};
        \draw[color=blue] (-2,0) arc (180:360:2cm);
        \fill (-2,0) circle (0.1);
        \fill (2,0) circle (0.1);
    \end{tikzpicture}
    \end{center}
    \label{fig:convEx1}
    \caption{Representation of $s$ in Eq.~\eqref{eq:sexample} illustrating Lemma~\ref{lem:CR2}; $\mathcal M_s$ is the open lower half circle (full line) and $\mathcal M_{\scr}$ is the closed lower half disc with the points $(1,0)$ and $(-1,0)$ removed.  }
\end{figure}
\begin{Lemma}\label{lem:CR2}
Let $\scr$ be the convex roof of $s : K_{ext}\mapsto \R$. Suppose $s$ has a minimum $m$  on $K_{\mathrm{ext}}$ and that
$m^{+}=\inf \{s(x)-m \mid x\in K_{ext} \backslash \mcl{M}_{s}\} > 0$. Then 
\[ 
\conv{\mcl{M}_{s}}\subset \mcl{M}_{\scr} \subset \overline{\conv{\mcl{M}_{s}}}
\]
\end{Lemma}
We note that if $\mathcal M_s$ is closed, then one finds, as in Proposition~\ref{prop:CR1}, that $\mcl{M}_{\scr} ={\conv{\mcl{M}_{s}}}$
\begin{proof}
As $s$ is bounded from below, we can suppose without loss of generality that $m= 0$ so that $s$ is nonnegative. We know from Proposition~\ref{prop:CR1} that $\conv{\mcl{M}_{s}}\subset \mcl{M}_{\scr}$. It remains to show the other inclusion. 
We want to prove that $\mcl{M}_{\scr} \subset \overline{\conv{\mcl{M}_{s}}}$. Thus, let x be in $\mcl{M}_{\scr}$ so that $\scr(x)=0$. By the definition of $\scr$, for all $ \epsilon > 0$, there exists $N_{\epsilon}\in\N^{*}$ such that

$$ \exists (\lambda_{i}^{\epsilon},x_i^{\epsilon})_{i\in \IntEnt{1}{N_{\epsilon}}}\in([0,1]\times K_{ext})^{N_{\epsilon}}, \ x=\sum_{i=1}^{N_{\epsilon}} \lambda_{i}^{\epsilon}x_i^{\epsilon}, \  0 = \scr(x) \leqslant \sum_{i=1}^{N_{\epsilon}} \lambda_{i}^{\epsilon}s(x_i^{\epsilon}) \leqslant \epsilon$$
and $\sum_{i=1}^{N_{\epsilon}}\lambda_{i}^{\epsilon} = 1$. Then,
\[ \epsilon \geqslant \sum_{i=1}^{N_{\epsilon}} \lambda_{i}^{\epsilon}s(x_i^{\epsilon}) \geqslant \sum_{\substack{i\in\IntEnt{1}{N_{\epsilon}} \\ x_{i}^{\epsilon}\notin \mcl{M}_{s}}} \lambda_{i}^{\epsilon}s(x_i^{\epsilon}) \geqslant m^{+} \sum_{\substack{i\in\IntEnt{1}{N_{\epsilon}} \\ x_{i}^{\epsilon}\notin \mcl{M}_{s}}} \lambda_{i}^{\epsilon}.
\]
And so,
 \begin{equation}\label{eq:CR2}
    \sum_{\substack{i\in\IntEnt{1}{N_{\epsilon}} \\ x_{i}^{\epsilon}\notin \mcl{M}_{s}}} \lambda_{i}^{\epsilon} \leqslant \frac{\epsilon}{m^{+}}.
\end{equation}
We define

\begin{equation}
    x^{\epsilon} = \frac{1}{P_{\epsilon}}\sum_{\substack{i\in\IntEnt{1}{N_{\epsilon}} \\ x_{i}^{\epsilon}\in \mcl{M}_{s}}} \lambda_{i}^{\epsilon}x_{i}^{\epsilon} \text{ where } P_{\epsilon} = \sum_{\substack{i\in\IntEnt{1}{N_{\epsilon}} \\ x_{i}^{\epsilon}\in \mcl{M}_{s}}} \lambda_{i}^{\epsilon}.
\end{equation}
Thus, for all $\epsilon >0$, $x_{\epsilon}\in\conv{\mcl{M}_{s}}$ and:
$$\left|\left|x - x^{\epsilon}\right|\right| \leqslant\sum_{\substack{i\in\IntEnt{1}{N_{\epsilon}} \\ x_{i}^{\epsilon}\in \mcl{M}_{s}}} \left|1-\frac{1}{P_{\epsilon}}\right|\lambda_{i}^{\epsilon}\norme{x_{i}^{\epsilon}}{} +  \sum_{\substack{i\in\IntEnt{1}{N_{\epsilon}} \\ x_{i}^{\epsilon}\notin \mcl{M}_{s}}} \lambda_{i}^{\epsilon}\norme{x_{i}^{\epsilon}}{} \leqslant \left|1-\frac{1}{P_{\epsilon}}\right|  A \sum_{\substack{i\in\IntEnt{1}{N_{\epsilon}} \\ x_{i}^{\epsilon}\in \mcl{M}_{s}}} \lambda_{i}^{\epsilon} +  A \sum_{\substack{i\in\IntEnt{1}{N_{\epsilon}} \\ x_{i}^{\epsilon}\notin \mcl{M}_{s}}} \lambda_{i}^{\epsilon}$$
where A is an upper bound of the norm of the elements of $K_{ext}$ which exists because $K$ is compact, hence  bounded. Finally, we obtain 
$$ \left|\left|x - x^{\epsilon}\right|\right| \leqslant \left| P_{\epsilon} -1\right| A + A \sum_{\substack{i\in\IntEnt{1}{N_{\epsilon}} \\ x_{i}^{\epsilon}\notin \mcl{M}_{s}}} \lambda_{i}^{\epsilon} \leqslant (1-P_{\epsilon}) A + A \sum_{\substack{i\in\IntEnt{1}{N_{\epsilon}} \\ x_{i}^{\epsilon}\notin \mcl{M}_{s}}} \lambda_{i}^{\epsilon} \leqslant 2 A \sum_{\substack{i\in\IntEnt{1}{N_{\epsilon}} \\ x_{i}^{\epsilon}\notin \mcl{M}_{s}}} \lambda_{i}^{\epsilon} \leqslant \frac{2A}{m^{+}}\epsilon$$
as $P_{\epsilon} +  \sum_{\substack{i\in\IntEnt{1}{N_{\epsilon}} \\ x_{i}^{\epsilon}\notin \mcl{M}_{s}}} \lambda_{i}^{\epsilon} =1$ by convex combination and thanks to the bound we had in \eqref{eq:CR2}. So $x^{\epsilon} \underset{\epsilon \rightarrow 0}{\rightarrow} x$ which shows that $x\in\overline{\conv{\mcl{M}_{s}}}$ and concludes the proof.
\end{proof}

\section{A counterexample to \eqref{eq:question2} in a spin-$1$ system}\label{App:spin1}
In~\cite{langrenez2023characterizing}, an example of a spin-$1$ system is given for which Eq.~\eqref{eq:sharpincl} holds:
\begin{equation*}
\Acal\cup\Bcal\subsetneq \EcalKDCpu\subsetneq \EcalKDCext.
\end{equation*}
For this spin-$1$ system, the minimum of $\nabcr$ equals $d+1=4$. We show here that, in that case, the implication in~\eqref{eq:question2}, which now reads
\begin{equation*}
\rho\in\EcalKDC\ \Rightarrow \ \{\nabcr(\rho)\leq d+1=\min \nabcr\} 
\end{equation*}
is violated. For that purpose, we  give an explicit example of a state $\rho$ that belongs to $\EcalKDC$ but for which $\nabcr(\rho)>d+1=4$.

First, we need to explicitly identify those pure states that have minimal support uncertainty $\nab$.

%%%%%%%%%%%%%%%
\subsection{Identification of the pure states in $\mcl{M}_{\nabcr}$}
In~\cite{langrenez2023characterizing}, we considered a spin-$1$ system taking $\Acal$ to be the $S_z=e_z\cdot \vec J$ basis and $\Bcal$ the $S_{z'}=e_{z'}\cdot \vec J$ basis, 
where  $e_{z'}=\frac13(2,2,-1)^T$.
%The $ \{ |b_m\rangle \}$ therefore forms, for $m=-1,0,1$,  the eigenbasis of the observable $e_{z'}\cdot J$ 
The transition matrix between these two bases is
\[
U = \frac{1}{3}\begin{pmatrix}
    -1 & 2 & 2 \\
    2 & -1 & 2 \\
    2 & 2 & -1
\end{pmatrix}.
\]
Since this system is completely incompatible, it is known~\cite{debievre2021,debievre2023a} that the pure  minimal uncertainty states  $|\psi\rangle$  satisfy $\nab(\psi)=d+1=4$. Using this information, explicit computations permit to identify these states to be the basis states $\Acal \cup \Bcal$ and, in addition, 9 other states:
\[
\cket{\phi_{1}} = \frac{\cket{a_2}-\cket{a_3}}{\sqrt{2}}, \ \cket{\phi_{2}} = \frac{\cket{a_1}-\cket{a_3}}{\sqrt{2}}, \cket{\phi_{3}} = \frac{\cket{a_1}-\cket{a_2}}{\sqrt{2}},
\]
\[
\cket{\psi_{1}} = \frac{\cket{a_1}+2\cket{a_2}}{\sqrt{5}},  \ \cket{\psi_{2}} = \frac{2\cket{a_1}+\cket{a_2}}{\sqrt{5}}, \ 
\cket{\psi_{3}} = \frac{\cket{a_1}+2\cket{a_3}}{\sqrt{5}},  
\]
\[
\cket{\psi_{4}} = \frac{2\cket{a_1}+\cket{a_3}}{\sqrt{5}}, \
\cket{\psi_{5}} = \frac{\cket{a_2}+2\cket{a_3}}{\sqrt{5}} \ \text{ and } \cket{\psi_{6}} = \frac{2\cket{a_2}+\cket{a_3}}{\sqrt{5}}.
\]
These $15$ pure states form the set $\mcl{M}_{\nab}$ and, according to Lemma~\ref{lem:CR2}, they are the extreme points of $\mcl{M}_{\nabcr}$. It was further shown in~\cite{langrenez2023characterizing} that all pure KD-positive states are given by  $\cket{\phi_1}$,$\cket{\phi_2}$, $\cket{\phi_3}$ and the basis states $\Acal\cup\Bcal$. In particular, it follows that  not all minimal support uncertainty pure states are KD-positive. 

We will prove in the next subsection that there exist mixed KD-positive states that do not have minimal uncertainty

\subsection{Identifying a KD-positive state not having minimal uncertainty.}

As $\Minnabcr=\conv{\mcl{M}_ {\nab}}$ and as there is a finite number of elements in $\mcl{M}_ {\nab}$, we can numerically exhaust all possible hyperplanes containing states in $\mcl{M}_ {\nab}$ and determine which ones are bounding hyperplanes of $\Minnabcr$. We find that there are $28$ bounding hyperplanes for $\Minnabcr$. However, we will study one particular bounding plane that is determined as follows. Consider the self-adjoint operator $F_{\star}$:
\begin{equation}\label{eq:Fstar}
    F_{\star} = \frac{1}{12}\begin{pmatrix}
        2 & -2 & -3 \\
        -2 & 6 & -1 \\
        -3 & -1 & 4
    \end{pmatrix}.
\end{equation}
Note that $F_{\star}$ has trace $1$.  Moreover, this operator determines a bounding plan for the convex set $\Minnabcr$ since any pure state $\cket{\psi}$ with minimal support uncertainty  satisfies the following inequality:
\begin{equation}\label{ex:spin1_traceF}
 \bra{\psi}F_{\star}\cket{\psi}=\mathrm{Tr}(F_{\star}\cket{\psi}\bra{\psi}) \leqslant \frac{1}{2}. 
\end{equation}
Thus, for any state $\rho$ in $\Minnabcr = \conv{\mcl{M}_{\nab}}$, we have that 
\begin{equation}\label{ex:spin1_traceFmixed}
 \mathrm{Tr}(F_{\star}\rho) \leqslant \frac{1}{2}. 
\end{equation}
This inequality is saturated when $\cket{\psi}$ is equal to one of $\cket{a_2},\cket{b_1},\cket{\phi_1},\cket{\phi_2}$ or $\cket{\phi_3}$. 
As a result, the affine hyperplane $\left\{\Tr(\rho F_\star)=\frac{1}{2}\right\}$ is a bounding hyperplane for $\Minnabcr$.
Note that
\begin{equation}\label{eq:TraceFsqu}
\mathrm{Tr}(F_{\star}^2) = \frac{7}{12} > \frac{1}{2}.
\end{equation}
One easily establishes, through an explict computation, that $Q(F_\star)$ only has real nonnegative entries, so that $F_\star$ is KD positive.

Now, we consider the following family of operators, for $\lambda\in[0,1]$:
\begin{equation}\label{eq:conv}
    \rho_{\lambda} = \lambda F_{\star} + \frac{(1-\lambda)}{3}(\cket{a_2}\bra{a_2} + \cket{b_1}\bra{b_1} + \cket{\phi_3}\bra{\phi_3}).
\end{equation}
For any $\lambda\in[0,1],$ $\rho_{\lambda}$ is a KD-positive operator as a convex combination of KD-positive operators and it has trace $1$. For $\rho_{\lambda}$ to be a state, we need to check if it is a positive operator.
We can compute the eigenvalues of $\rho_{\lambda}$:
\begin{align}
    \nonumber r_{1}(\lambda) &=\frac{1}{18}\left( 7\lambda +2 \right) ,\\
    \nonumber r_{2}(\lambda) &= \frac{1}{36}\left(\sqrt{7\lambda^2 - 32\lambda +160} -7\lambda +16\right) ,\ \\
    r_{3}(\lambda) &=\frac{1}{36}\left(-\sqrt{7\lambda^2 - 32\lambda +160} -7\lambda +16\right). 
\end{align}
These eigenvalues are all real and nonnegative if and only if   
\[
 -7\lambda +16 \geqslant \sqrt{7\lambda^2 - 32\lambda +160}.
\]
This inequality is satisfied for $\lambda\in[0,\frac{4}{7}]$  implying that for these values $\rho_{\lambda}$ is a KD-positive state. Moreover, we have that 
\[
\mathrm{Tr}(F_{\star}\rho_{\lambda}) = \frac{1}{2} + \frac{\lambda}{12}.
\]
Consequently, for $\lambda\in]0,\frac{4}{7}]$, $\rho_{\lambda}$ is a KD-positive state satisfying $\mathrm{Tr}(F_{\star}\rho_{\lambda}) > \frac{1}{2}$. Thus, for this range of $\lambda$ values, $\rho_{\lambda}$ does not belong to $\Minnabcr$, meaning that $\rho_{\lambda}$ satisfies 
\[
\nabcr(\rho_{\lambda}) > 4 =d+1.
\]
Thus, we have an example in dimension $3$ for which there exist mixed KD-positive states that do not have minimal support uncertainty, meaning that in this case, \eqref{eq:question2} does not hold.

\newpage

\bibliographystyle{ieeetr}
\bibliography{BIB_KD}

\end{document}